\documentclass[11pt]{article}

\usepackage{amssymb,amsmath,amsfonts}
\usepackage{graphicx,color,enumitem}
\usepackage{amsthm} 
\usepackage{bm}
\usepackage[round]{natbib}

\usepackage{geometry}

\usepackage[colorinlistoftodos, textwidth=4cm, shadow]{todonotes}
\RequirePackage[colorlinks,citecolor=blue,urlcolor=blue]{hyperref}

\usepackage{tikz}
\tikzstyle{vertex}=[circle, draw, inner sep=2pt, fill=white]

\newcommand{\E}{{\mathbb E}}

\newcommand{\C}{{\mathbb C}}
\newcommand{\R}{{\mathbb R}}

\newcommand{\Fcal}{{\mathcal F}}

\newcommand{\Rcal}{{\mathcal R}}

\newcommand{\Mid}{{\ \Big|\ }}

\newcommand{\QVar}[1]{\left\langle{#1}\right\rangle}
\newcommand{\QCov}[2]{\left\langle{#1},{#2}\right\rangle}

\newcommand{\1}{\mathbf{1}}

\newcommand{\fdot}{{\,\cdot\,}}

\newtheorem{theorem}{Theorem}

\newtheorem{example}[theorem]{Example}

\newtheorem{proposition}[theorem]{Proposition}
\newtheorem{remark}[theorem]{Remark}

\numberwithin{equation}{section}
\numberwithin{theorem}{section}

\begin{document}

\title{Affine Rough Models\footnote{The authors would like to thank Eduardo Abi Jaber and Christa Cuchiero for valuable discussions and suggestions. Martin Keller-Ressel gratefully acknowledges financial support from DFG grants ZUK~64 and KE~1736/1-1. Martin Larsson gratefully acknowledges financial support from SNF Grant 205121\_163425. Sergio Pulido gratefully acknowledges financial support from MATH AmSud project SaSMoTiDep 18-MATH-17.}}
\author{Martin Keller-Ressel\thanks{Institute of Mathematical Stochastics, TU Dresden, 01062 Dresden,
Germany, martin.keller-ressel@tu-dresden.de.} \and Martin Larsson\thanks{Department of Mathematics, ETH Zurich, R\"amistrasse 101, CH-8092, Zurich, Switzerland, martin.larsson@math.ethz.ch.} \and Sergio Pulido\thanks{Laboratoire de Math\'ematiques et Mod\'elisation d'\'Evry (LaMME), Universit\'e d'\'Evry-Val-d'Essonne, ENSIIE, Universit\'e Paris-Saclay, UMR CNRS 8071, IBGBI 23 Boulevard de France, 91037 \'Evry Cedex, France, sergio.pulidonino@ensiie.fr.}}
\date{December 20, 2018}

\maketitle
\begin{abstract}
The goal of this survey article is to explain and elucidate the affine structure of recent models appearing in the rough volatility literature, and show how it leads to exponential-affine transform formulas.
\end{abstract}

\section{Introduction}

Affine stochastic volatility models have a long history in the quantitative finance literature; see e.g.\ \cite{duf_fil_sch_03,Kallsen:2006} and the references listed there. These models are generally of the form
\begin{equation}\label{eq:S_SDE}
dS_t = S_t \sqrt{V_t} dB_t,
\end{equation}
where $S$ is the asset price, and the (spot) variance $V$ is modeled by an affine process. Arguably, the most prominent example is the \cite{heston1993closed} model, where $V$ follows a scalar square-root diffusion. The affine property leads to tractable Fourier--Laplace transforms of various quantities of interest. For example, the log-price satisfies the exponential-affine transform formula
\begin{equation}\label{eq:affine}
\E[\exp(v \log S_T) \mid \Fcal_t] = \exp\left(v \log S_t + \phi(T-t) + \psi(T-t) V_t\right),
\end{equation}
where $\phi, \psi$ are the solutions to ordinary differential equations of Riccati type that depend on $v$. Similar formulas also exist for the spot variance and integrated spot variance.

Unfortunately, these models do not produce the rough trajectories of volatility that seem to occur empirically, see \cite{gatheral2018volatility}, and have trouble capturing the term structure of implied volatilities and its skew, cf. \cite{fukasawa2017short}. Still, it is possible to construct stochastic volatility models with these features, and with an ``affine structure'' that produces formulas similar to \eqref{eq:affine}. This has recently been done by \cite{GJR:14,euch2016characteristic,jaber2017affine,gatheral2018affine}, and related ideas appear already in \cite{Comte2012}. The goal of this chapter is to explain and elucidate this ``affine structure'', and show how it leads to exponential-affine transform formulas.

We will give four perspectives. The point of departure, in Section~\ref{sec:SCE}, is a class of {\em stochastic convolution equations} for $V$ with affine coefficients, which contains the rough Heston model of \cite{euch2016characteristic} as an immediate special case. This is the first perspective. The second perspective, in Section~\ref{sec_fwdvar}, is to view these models as {\em forward variance models}, which focus on the forward variance curve
\begin{equation}\label{eq:fw_var}
\xi_t(T) = \E[V_T \mid \Fcal_t].
\end{equation}
The third perspective, in Section~\ref{sec:forwardrep}, is to regard a modified forward variance curve as the solution of a {\em stochastic partial differential equation}. Finally, the fourth perspective, in Section~\ref{sec_Laplace}, is available when the convolution kernel is the Laplace transform of a possibly signed measure. This leads to a representation as a {\em mixture of mean-reverting processes}. Dually, this gives multiple perspectives on the Riccati equations that characterize the associated Fourier--Laplace functionals.

Space constraints prevent us from including rigorous proofs of all results. Still, some proofs and derivations are presented, selected because they are instructive without being too long. We occasionally use the convolution notation $(f*g)(t)=\int_0^t f(t-s)g(s)ds$ for functions $f$ and $g$, and similarly $(f*dZ)_t=\int_0^t f(t-s)dZ_s$ when $Z$ is a semimartingale.

\section{Stochastic convolution equations}
\label{sec:SCE}

Consider the stochastic convolution equation
\begin{equation} \label{eq_SCE}
V_t = V_0 + \int_0^t K(t-s) b(V_s)ds + \int_0^t K(t-s) \sigma(V_s) dW_s,
\end{equation}
for some real continuous functions $b$ and $\sigma$, kernel $K\in L^2_{\rm loc}(\R_+)$, initial condition $V_0\in\R$, and Brownian motion $W$. Solutions to \eqref{eq_SCE} are always understood to have continuous paths.

\begin{example}\label{ex_1}
Taking $b(x)=\lambda(\theta-x)$, $\sigma(x)=\zeta\sqrt{x}$, and the power-law kernel $K_\text{$\alpha$-pow}(t) = t^{\alpha-1} / \Gamma(\alpha)$ with $\alpha\in(\frac12,1)$, we obtain the spot variance process in the rough Heston model of \cite{euch2016characteristic}. With $\alpha=1$, we recover the spot variance process in the classical Heston model.
\end{example}

Our focus is on the case where $V$ is an {\em affine Volterra process}, which is when $b(x)$ and $\sigma(x)^2$ are affine in $x$. This definition naturally generalizes to higher dimension: $b(x)$ is then a vector and $\sigma(x)$ a matrix, and one requires $b(x)$ and $\sigma(x)\sigma(x)^\top$ to be affine in $x$. In this chapter we focus on the one-dimensional affine case, so that
\begin{equation}\label{eq_b(x)a(x)}
\text{$b(x)=\beta-\lambda x$ and $\sigma(x)^2=\alpha+ax$}
\end{equation}
for some real parameters $\beta,\lambda,\alpha,a$ such that $\alpha+aV_t\ge0$ for all $t\ge0$. The latter condition raises delicate questions of existence of solutions to \eqref{eq_SCE} when $a\ne0$, which we do not address here in detail. Let us however state the following result, whose proof can be found in \cite{jaber2017affine}. Part~\ref{T_existence:2} of the theorem requires the following assumption on the kernel:
\begin{equation}\label{eq:K_cond}
\begin{minipage}[c][2.5em][c]{.85\textwidth}
$K$ is strictly positive and completely monotone. There is $\gamma \in (0,2]$ such that $\int_0^h K(t)^2dt = O(h^\gamma)$ and $\int_0^T (K(t+h) - K(t))^2 dt = O(h^\gamma)$ for every $T < \infty$.
\end{minipage}
\end{equation}
(Recall that a $C^\infty$ function $f\colon(0,\infty)\to\R$ is completely monotone if $(-1)^k f^{(k)}\ge0$, $k\ge0$.) In particular, the power-law kernel in Example~\ref{ex_1} satisfies \eqref{eq:K_cond} with $\gamma=2\alpha-1$.

\begin{theorem}\label{T_existence}
Consider the equation \eqref{eq_SCE} with coefficients $b(x)$ and $\sigma(x)$ as in \eqref{eq_b(x)a(x)} and kernel $K\in L^2_{\rm loc}(\R_+)$.
\begin{enumerate}
\item\label{T_existence:1} Assume that $\alpha\ge0$ and $a=0$. Then there exists a pathwise unique strong solution $V$ for any initial condition $V_0\in\R$; the \textbf{Volterra Ornstein--Uhlenbeck process}.
\item\label{T_existence:2} Assume that $\alpha=0$, $a>0$, $\beta\ge0$, and $K$ satisfies \eqref{eq:K_cond}. Then there exists a unique in law $\R_+$-valued weak solution $V$ for any initial condition $V_0\in\R_+$; the \textbf{Volterra square-root process}.
\end{enumerate}
In either case, the trajectories of $V$ are H\"older continuous of any order less than $\gamma/2$.
\end{theorem}

\begin{remark}
A solution $V$ is called strong if it is adapted to the filtration generated by the Brownian motion $W$. This is not required for a weak solution, where one is free to construct the Brownian motion as needed. Pathwise uniqueness means that any two solutions $V$ and $V'$ to \eqref{eq_SCE} driven by the same Brownian motion must have identical trajectories (outside a nullset). It is not known whether pathwise uniqueness holds in \ref{T_existence:2}.
\end{remark}

A Volterra--Heston model is a stochastic volatility model of the form \eqref{eq:S_SDE}, where the spot variance $V$ is a Volterra square-root process. Most of this chapter is concerned with such models. The process $V$ is generally neither a Markov process nor a semimartingale. This causes difficulties that the alternative perspectives developed in the following sections help to circumvent.

\section{Forward variance models}\label{sec_fwdvar}

A useful perspective on \eqref{eq_SCE} is as a {\em forward variance model}, as noted e.g.\ by \cite{buehler2006volatility, ber_guy_12} and \cite{bay_fri_gat_16}. Consider a model \eqref{eq:S_SDE} where the spot variance process $V$ is given by the stochastic convolution equation \eqref{eq_SCE} with affine drift $b(x) = \lambda(\theta - x)$, i.e. 
\begin{equation}\label{eq:V}
V_t = V_0 + \lambda \int_0^t K(t-s)(\theta-V_s)ds + \int_0^t K(t-s) \sigma(V_s)dW_s.
\end{equation}
Our first goal is to derive an SDE for the forward variance $\xi_t(T)$ defined in \eqref{eq:fw_var}. To this end, we remark that for any kernel $k\in L^2_{\rm loc}(\R_+)$ there exists a unique kernel $r\in L^2_{\rm loc}(\R_+)$, called the {\em resolvent} or {\em resolvent of the second kind} of $k$, such that
\[
k(t) - r(t) = \int_0^t r(t-s)k(s)ds, \quad t\ge0.
\]

\begin{example}
If $k(t)\equiv c$ is constant, then $r(t)=ce^{-ct}$. If $k(t)=c\,t^{\alpha-1}/\Gamma(\alpha)$ is proportional to the power-law kernel, then $r(t)=ct^{\alpha-1}  E_{\alpha, \alpha} (-ct^{\alpha})$ where $E_{\alpha, \alpha}$ denotes the Mittag-Leffler function.
\end{example}

The forward variance dynamics can now be described as follows.

\begin{proposition}\label{P_fwdvar}
Let $R_\lambda$ be the resolvent of $\lambda K$. The forward variance $\xi_t(T)$ associated to \eqref{eq:V} satisfies
\[
d \xi_t(T) = \tfrac{1}{\lambda} R_\lambda(T-t) \sigma(V_t) dW_t
\]
with initial condition 
\[
\xi_0(T) = V_0 \left(1 - \int_0^T R_{\lambda}(s) ds \right) + \theta \int_0^T R_\lambda(s) ds. 
\]
If $\lambda=0$, interpret $\lambda^{-1}R_\lambda=K$, and note that $R_\lambda=0$ in this case.
\end{proposition}

\begin{proof}
Suppose $\lambda\ne0$; otherwise, the proof is easier and does not use resolvents. Denote by $\bf{1}$ the function which takes the constant value $1$. The spot variance process $V$ is given by $V = V_0 + \lambda K * (\theta - V)  + K * (\sigma(X)dW)$. Therefore,
\[
V-R_\lambda*V = V_0 (1 - R_\lambda * \bm{1}) + \lambda (K - R_\lambda* K)  * (\theta  - V) + (K - R_\lambda * K) * (\sigma(V)dW).
\]
By definition of the resolvent, $K - R_\lambda* K = \tfrac{1}{\lambda} R_\lambda$. Plug this in and cancel the $-R_\lambda*V$ terms on both sides to get
\begin{equation}\label{eq:V_forward}
V = V_0 (1 - R_\lambda * \bm{1}) + \theta R_\lambda * \bm{1}  + \tfrac{1}{\lambda} R_\lambda * (\sigma(V)dW).
\end{equation}
The process $M_u := \int_0^u R_\lambda(T-s) \sigma(X_s) dW_s$, $u\in[0,T]$, is a martingale. Therefore, evaluating \eqref{eq:V_forward} at $T$ and taking $\Fcal_t$-conditional expectations yields
\[
\xi_t(T) = \E[V_T \mid \Fcal_t] = \xi_0(T) + \int_0^T \tfrac{1}{\lambda} R_\lambda(T-s) \sigma(V_s) dW_s,
\]
which is the claimed result.
\end{proof}

\begin{remark}
We ignored some technical but important points in the proof. First, the associativity property $(k_1*k_2)*dZ=k_1*(k_2*dZ)$ was used for certain kernels $k_1$, $k_2$ and $dZ=\sigma(X)dW$. This identity can be proved using the stochastic Fubini theorem. Second, we did not verify that $M$ is really a martingale, not just a local martingale. For $\sigma(x) = \alpha + a x$ this can be done by noting that
\[
\E[\langle M\rangle_T] = \tfrac{1}{\lambda^2}\int_0^T R_\lambda(T-s)^2\, \E[\sigma(X_s)^2] ds \le C(1+\sup_{s\le T}\E[|X_s|]),
\]
where one can take $C=\tfrac{1}{\lambda^2}(|\alpha|+|a|)\int_0^T R_\lambda(s)^2 ds$. The right-hand side is finite, so $M$ is actually a square-integrable martingale. Details are given by \cite{jaber2017affine}.
\end{remark}

\subsection{Fourier--Laplace transforms and Riccati--Volterra equations}

In the affine case \eqref{eq_b(x)a(x)}, not only conditional expectations have useful representations, but also Fourier--Laplace transforms. We now explain how such representations can be derived once the Volterra--Heston model (in log-price notation $L = \log S$) is written in forward variance form,
\begin{equation}\label{eq:Volterra_Heston}
\left\{ \
\begin{aligned}
dL_t &= -\tfrac{1}{2}V_t dt + \sqrt{V_t} dB_t\\
d\xi_t(T) &= \tfrac{1}{\lambda} R_\lambda(T-t) \sigma\sqrt{V_t} dW_t.\\
\end{aligned}
\right.
\end{equation}
Here $\lambda \ge 0$, $d\langle B,W\rangle_t = \rho dt$ for $\rho\in[0,1]$, and the diffusion part has the affine form \eqref{eq_b(x)a(x)} with $\alpha=0$ and $a=\sigma^2$ for some $\sigma>0$. Define the function 
\begin{equation}\label{eq:Q_def}
Q(u,z) = \frac{1}{2}(u^2 - u) + \sigma \rho u z + \frac{\sigma^2}{2}z^2, \quad u,z \in \C.
\end{equation}

\begin{theorem}\label{thm:Riccati_resolvent}
Consider the Volterra--Heston model \eqref{eq:Volterra_Heston}. Fix $T > 0$ and $(u,v,w) \in \C^3$, and assume that the Riccati--Volterra equation
\begin{equation}\label{eq:Riccati_kernel}
\psi = vK + K * \left(Q(u,\psi) - \lambda \psi + w\right)
\end{equation}
has a solution $\psi\in L^2(0,T)$. Then the auxiliary process
\begin{equation}\label{eq:M_def}
M_t = \exp\left( u L_t + v\xi_t(T) + w\int_0^T\xi_t(s)ds + \int_t^T \xi_t(s) Q(u,\psi(T-s)) ds \right)
\end{equation}
is a local martingale on $[0,T]$, and satisfies
\begin{equation}\label{eq_dMM_psi_NEW}
\frac{dM_t}{M_t} = u\sqrt{V_t}dB_t +  \psi(T-t) \sigma \sqrt{V_t} dW_t.
\end{equation}
If $M$ is a true martingale, the joint conditional Fourier--Laplace transform of the triplet $(L_T,V_T,\int_0^T V_sds)$ is $\E[\exp(u L_T+vV_T+w\int_0^TV_sds) \mid \Fcal_t]=M_t$.
\end{theorem}

The two crucial assumptions are of course that \eqref{eq:Riccati_kernel} has a solution, and that the local martingale $M$ is really a true martingale. The following theorem gives a sufficient condition that guarantees this; the proof can be found in \cite{jaber2017affine}.

\begin{theorem}
Let $K$ be a kernel satisfying \eqref{eq:K_cond} and let $(u,v,w) \in \C^3$ satisfy ${\rm Re\,}u\in[0,1]$, ${\rm Re\,}v\le0$, and ${\rm Re\,}w\le0$. Then the Riccati--Volterra equation \eqref{eq:Riccati_kernel} has a unique global solution $\psi$, and the local martingale $M$ in \eqref{eq:M_def} is a true martingale.
\end{theorem}

We now present the proof of Theorem~\ref{thm:Riccati_resolvent} in the special case where $v=w=0$. This simplifies the calculations, and the proof in the general case is similar.

\begin{proof}[Proof of Theorem~\ref{thm:Riccati_resolvent} for $v=w=0$]
Subtract $R_\lambda * \psi$ from both sides of \eqref{eq:Riccati_kernel}, where now $v=w=0$, and apply the resolvent equation $K - R_\lambda* K = \tfrac{1}{\lambda}R_\lambda$ to get an equivalent form of the Riccati--Volterra equation,
\begin{equation}\label{eq:psi_resolvent}
\psi = \tfrac{1}{\lambda}R_\lambda * Q(u,\psi). 
\end{equation}
We aim to apply It\^o's formula to $M_t$, so we define $G_t = \int_t^T Q(u,\psi(T-s))\xi_t(s) ds$. Since $\xi_t(s)=V_s$ for $s\le t$, we can write
\[
G_t = \int_0^T Q(u,\psi(T-s))\xi_t(s)ds - \int_0^t Q(u,\psi(T-s))V_sds.
\]
Focus on the first term. Using first Proposition~\ref{P_fwdvar}, then the stochastic Fubini theorem \cite[Thm.~2.2]{veraar2012stochastic}, and finally \eqref{eq:psi_resolvent}, we get
\begin{align*}
\int_0^T Q&(u,\psi(T-s))\xi_t(s)ds \\
&= \int_0^T Q(u,\psi(T-s))\left\{\xi_0(s) + \sigma\int_0^{t\wedge s} \tfrac{1}{\lambda} R_\lambda(s-r)\sqrt{V_r} dW_r\right\}ds \\
&= \int_0^T Q(u,\psi(T-s)) \xi_0(s) ds +  \sigma \int_0^t \int_r^T  Q(u,\psi(T-s)) \tfrac{1}{\lambda}  R_\lambda(s-r) ds \sqrt{V_r} dW_r \\
&= \int_0^T Q(u,\psi(T-s)) \xi_0(s) ds +  \sigma \int_0^t \psi(T-r) \sqrt{V_r} dW_r.
\end{align*}
As a result,
\[
G_t = \int_0^T Q(u,\psi(T-s)) \xi_0(s) ds +  \sigma \int_0^t \psi(T-r) \sqrt{V_r} dW_r  - \int_0^t Q(u,\psi(T-s))V_sds.
\]
This leaves us in a position to apply It\^o's formula to $M_t = \exp \left(u L_t + G_t \right)$, giving
\begin{align*}
\frac{dM_t}{M_t} &= u\, dL_t + dG_t + \frac{u^2}{2} d\QVar{L}_t + u\, d\QCov{L}{G}_t + \frac{1}{2} d\QVar{G}_t \\
&= \left\{ \frac{1}{2}(u^2 - u) - Q(u,\psi(T-t)) + u \rho \sigma \psi(T-t) + \frac{\sigma^2}{2} \psi(T-t)^2 \right\} V_t\,dt \\
&\quad + u\sqrt{V_t}dB_t +  \psi(T-t) \sigma \sqrt{V_t} dW_t.
\end{align*}
Comparing with \eqref{eq:Q_def} shows that the $dt$-term vanishes, so $M$ is indeed the local martingale in \eqref{eq_dMM_psi_NEW}. If $M$ is a true martingale, we conclude that $\E[\exp(uL_T)\mid \Fcal_t] = \E[M_T \mid \Fcal_t] = M_t$, as claimed.
\end{proof}

\begin{remark}
Setting $g = Q(u,\psi)$ and applying $Q(u,\fdot)$ to both sides of \eqref{eq:psi_resolvent} yields
\[g = Q(u,\tfrac{1}{\lambda} R_\lambda * g).\]
This is the `convolution Riccati equation' considered by \cite{gatheral2018affine}, which leads to an equivalent formulation in terms of $g$ instead of $\psi$. 
\end{remark}

\subsection{Necessity of affine representations}

We now discuss a converse to Theorem~\ref{thm:Riccati_resolvent} obtained by \cite{gatheral2018affine}. Consider a general forward variance model of the type
\[d\xi_t(T) = \eta_t(T) dW_t,\]
where $\eta_t(T)$ is decreasing in $T$ and $\xi_t(T)$ is the forward variance associated to a price process $S$ of the form $dS_t = S_t a(V_t) dW_t$. Assume that the conditional cumulant generating function of the log-price $L_T = \log S_T$ is of the form 
\[
\E[\exp(u L_T) \mid \Fcal_t] = \exp \left( u L_t + \int_t^T \xi_t(T-s) g(s,u) ds \right)
\]
for all $u \in [0,1]$ and $0 \le t \le T$, for some continuous function $g \le 0$. Under mild integrability conditions on $\eta_t(T)$, \cite{gatheral2018affine} then show that, necessarily,
\[
a(V_t) = a \sqrt{V_t}\quad \text{and} \quad \eta_t(T) = \kappa(T-t) \sqrt{V_t}
\]
for some constant $a\ge0$ and kernel $\kappa$. Thus the model is precisely of the form \eqref{eq:Volterra_Heston}, and $\kappa$ is identified with the resolvent $\tfrac{1}{\lambda}R_\lambda$.

\subsection{Fractional calculus and the rough Heston model}

Consider the power law kernel $K_\text{$\alpha$-pow}(t) = t^{\alpha-1}/\Gamma(\alpha)$ used in the rough Heston model. The {\em Riemann--Liouville fractional integral} $I^\alpha$ is defined via convolution with this kernel, $I^\alpha f=K_\text{$\alpha$-pow}*f$. One then defines the Riemann--Liouville {\em fractional derivative} $D^\alpha$ as $D^\alpha f = \tfrac{d}{dt} I^{1 - \alpha} f$, which provides an inverse to the fractional integral in that $D^\alpha(I^\alpha f)=I^\alpha(D^\alpha f)=f$. It follows that, in the case $v=w=0$, \eqref{eq:Riccati_kernel} is equivalent to
\[D^\alpha \psi = Q(u,\psi) - \lambda \psi,\]
which is precisely the fractional Riccati equation derived by \cite{euch2016characteristic}. Using Proposition~\ref{P_fwdvar} and \eqref{eq:psi_resolvent} we can rewrite the exponent in \eqref{eq:M_def}, for $t = 0$, as
\begin{align*}
\xi_0 * Q(u,\psi) &= V_0\, \bm{1} * Q(u,\psi) + (\theta - V_0) \left(\bm{1} * R_\lambda * Q(u,\psi)\right) \\
&= V_0\, \bm{1} * Q(u,\psi) + \lambda (\theta - V_0) (\bm{1} * \psi) \\
&= V_0\, I^{1-\alpha} \psi + \lambda \theta \left(\bm{1} * \psi\right).
\end{align*}
Thus, in the rough Heston model, the unconditional transform formula in Theorem~\ref{thm:Riccati_resolvent}, with $v=w=0$, becomes
\[\E[\exp (u L_T)] = \exp \left( u L_0 + \lambda \theta \int_0^T \psi(s)ds + V_0\, I^{1-\alpha} \psi(T)\right),\]
which is consistent with \cite{euch2016characteristic}.\\

\section{Modified forward process representation}\label{sec:forwardrep}

Another perspective on \eqref{eq_SCE} via a {\em stochastic partial differential equation} arises as follows.
Starting with a Volterra process $V$ of the form~\eqref{eq_SCE}, define the process
\begin{equation*}
u_t(x)=\E\left[V_{t+x}-\int_t^{t+x}K(t-s+x)b(V_s)\,ds\Mid\mathcal F_t\right].
\end{equation*}
This process is considered by \cite{jaberelouch2018markovaffineroughHeston}. We call it the {\it modified forward process}, because had we not subtracted the time integral, we would have obtained the so-called Musiela parameterization $\xi_t(t+x)$ of the forward process. The only term inside the conditional expectation that is not already $\Fcal_t$-measurable is an integral with respect to $W$. This gives
\begin{equation}\label{eq_utx}
u_t(x)=V_0+\int_0^t K(t-s+x)b(V_s)\,ds+\int_0^t K(t-s+x)\sigma(V_s)\,dW_s,
\end{equation}
which can be expressed in terms of the following SPDE.

\begin{proposition}\label{prop_SPDE_SCE}
The process $u_t(x)$ in \eqref{eq_utx} is a mild solution of the SPDE
\begin{equation}\label{eq_utx2}
du_t(x)=(\partial_x u_t(x)+K(x)b(u_t(0)))dt+K(x)\sigma(u_t(0))dW_t
\end{equation}
with initial condition $u_0(x)=V_0$ for all $x$.
\end{proposition}

\begin{proof}
Formally taking the differential in \eqref{eq_utx}, using that $\partial_t K(t-s+x)=\partial_x K(t-s+x)$ and that $u_t(0)=V_t$, gives \eqref{eq_utx2}. More rigorously, note that $K(t-s+x)=T_{t-s}K(x)$, where $T_{t-s}$ is the shift operator that maps any function $f$ to the shifted function $f(t-s+\fdot)$. The derivative $\partial_x$ is the infinitesimal generator of the shift semigroup $\{T_t\}_{t\ge0}$, so, by definition, \eqref{eq_utx} is actually the mild formulation of the SPDE \eqref{eq_utx2}; see \cite[Section~6.1]{da2014stochastic}.
\end{proof}

\subsection{Fourier--Laplace transforms and Riccati equations}

The SPDE \eqref{eq_utx2} suggests that the process $\{u_t(\fdot)\}_{t\ge 0}$ is an infinite dimensional Markov process. In the affine case \eqref{eq_b(x)a(x)}, we therefore expect a Fourier--Laplace transform formula like
\begin{equation}\label{E:chfunctliftcurve}
\E\left[\exp\left(\int_0^\infty h(x) u_T(x)dx\right)\Mid\mathcal{F}_t\right]=\exp\left(\phi(T-t)+\int_0^\infty\Psi(T-t,x)u_t(x)dx\right),
\end{equation}
where $\phi(\tau)$ and $\Psi(\tau,x)$ are solutions of appropriate Riccati equations. These equations turn out to be
\begin{align}
\partial_t\phi(t)&= \mathcal R_{\phi}\left(\textstyle\int_0^\infty \Psi(t,y)K(y)dy\right) \label{E:Riccatilift1_forward_mild}\\
\Psi(t,x)&=h(x-t)\1_{\{x\geq t\}}+\mathcal R_{\Psi}\left(\textstyle\int_0^\infty \Psi(t-x,y)K(y)dy\right)\1_{\{x<t\}\label{E:Riccatilift2_forward_mild}}
\end{align}
with $\phi(0)=0$ and where we define
\begin{equation}
\label{eq:RphiRpsi}
\Rcal_{\phi}(y)=\beta y+\frac{\alpha}{2}y^2,\qquad \Rcal_{\Psi}(y)=-\lambda y+\frac{a}{2}y^2.
\end{equation}

\begin{remark}
At first sight, \eqref{E:Riccatilift2_forward_mild} does not look like a differential equation for $\Psi(t,x)$. But, along the lines of the proof of Proposition~\ref{prop_SPDE_SCE}, \eqref{E:Riccatilift2_forward_mild} can actually be viewed as a mild formulation of the formal PDE
\[
\partial_t\Psi(t,x) = -\partial_x\Psi(t,x)+\mathcal R_{\Psi}\left(\int_0^{\infty}\Psi(t,y) K(y)\,dx\right)\delta_0(x)
\]
with initial condition $\Psi(0,x)=h(x)$.
\end{remark}

Let us give a derivation of the Riccati equations \eqref{E:Riccatilift1_forward_mild}--\eqref{E:Riccatilift2_forward_mild}. We assume that $V_0=0$; this does not affect the validity of the Riccati equations, but simplifies the calculations. Suppose that $\Psi(t,x)$ satisfies~\eqref{E:Riccatilift2_forward_mild} and define $dZ_t=b(V_t)dt+\sigma(V_t)dW_t$, a semimartingale. Using \eqref{eq_utx}, \eqref{E:Riccatilift2_forward_mild}, and the stochastic Fubini theorem; then a change of variables; and finally \eqref{E:Riccatilift2_forward_mild} once again, we get
\begin{align*}
\int_0^{\infty}&\Psi(T-t,x)u_t(x)dx\\
&=\int_0^t \int_{T-t}^{\infty} h(x-T+t)K(t-s+x)\,dx\,dZ_s\\
&\quad+\int_0^t \int_{0}^{T-t} \mathcal R_{\Psi}\left(\textstyle\int_0^\infty \Psi(T-t-x,z)K(z)dz\right)K(t-s+x)\,dx\,dZ_s\\
&=\int_0^t \int_{T-s}^{\infty} h(y-T+s)K(y)\,dy\,dZ_s\\
&\quad+\int_0^t \int_{t-s}^{T-s} \mathcal R_{\Psi}\left(\textstyle\int_0^\infty \Psi(T-s-y,z)K(z)dz\right)K(y)\,dy\,dZ_s\\
&=\int_0^t \int_0^\infty \Psi(T-s,x)K(x)dx\,dZ_s \\
&\quad-\int_0^t\int_0^{t-s} \mathcal R_{\Psi}\left(\textstyle\int_0^\infty \Psi(T-s-y,z)K(z)dz\right)K(y)\,dy\,dZ_s.
\end{align*}
Combining this with the stochastic Volterra equation~\eqref{eq_SCE} satisfied by $V$ yields
\begin{align*}
d\int_0^{\infty}&\Psi(T-t,x)u_t(x)dx\\
&=\int_0^\infty \Psi(T-t,x)K(x)dx\,dZ_t - \int_0^t R_{\Psi}\left(\textstyle\int_0^\infty \Psi(T-t,y)K(y)dy\right)K(t-s)\,dZ_s\,dt\\
&=\int_0^\infty \Psi(T-t,x)K(x)dx\,dZ_t -R_{\Psi}\left(\textstyle\int_0^\infty \Psi(T-t,y)K(y)dy\right)V_t\,dt.
\end{align*}
Let $M_t$ denote the right-hand side of~\eqref{E:chfunctliftcurve}. Use the previous equation and \eqref{E:Riccatilift1_forward_mild} to get
\begin{equation}\label{eq_dMM_sec5}
\frac{dM_t}{M_t}=\int_0^\infty \Psi(T-t,x)K(x)dx\,\sigma(V_t)\,dW_t.
\end{equation}
Thus $M$ is a local martingale, and $M_T=\exp(\int_0^\infty h(x)u_T(x)dx)$ since $\Psi(0,x)=h(x)$. If $M$ is a true martingale we deduce the exponential-affine formula~\eqref{E:chfunctliftcurve}.

This can be used to derive the special case of Theorem~\ref{thm:Riccati_resolvent} where $u=w=0$ (note that $\alpha=0$ and $a=\sigma^2$ in that theorem). Formally setting $h=v\delta_0$ with $v\in\mathbb C$ gives
\[
\E\left[\exp\left( v V_T \right)\mid\mathcal{F}_t\right]=\exp\left(\phi(T-t)+\int_0^{\infty}\Psi(T-t,x)u_t(x)\,dx\right).
\]
There is a connection between the Riccati equation \eqref{E:Riccatilift2_forward_mild} and the Riccati--Volterra equation \eqref{eq:Riccati_kernel} in Theorem~\ref{thm:Riccati_resolvent}. To wit, suppose that $\Psi(t,x)$ solves \eqref{E:Riccatilift2_forward_mild} and define
\begin{equation}\label{eq_psi_eq_sec5}
\psi(t)=\int_0^{\infty} \Psi(t,x)K(x)dx.
\end{equation}
Using the definition \eqref{eq:RphiRpsi} of $\Rcal_\Psi$, the definition \eqref{eq_psi_eq_sec5} of $\psi$, \eqref{E:Riccatilift2_forward_mild}, and a change of variables, we get
\begin{align*}
K\ast\left(-\lambda \psi+\frac{a}{2}\psi^2\right)(t)
&=\int_0^\infty K(x)\mathcal R_{\Psi}(\psi(t-x))\bm1_{\{x<t\}}dx\\
&=\int_0^\infty K(x)(\Psi(t,x)-h(x-t)\1_{\{x\geq t\}})dx\\
&=\psi(t)-\int_0^{\infty}h(x)K(t+x)dx.
\end{align*}
If $h=v\delta_0$, we deduce the Riccati--Volterra equation \eqref{eq:Riccati_kernel} for the case $u=w=0$. Observe also that, in view of \eqref{eq_psi_eq_sec5}, \eqref{eq_dMM_sec5} agrees with \eqref{eq_dMM_psi_NEW}.

\section{Laplace representation}\label{sec_Laplace}

Our final perspective on \eqref{eq_SCE} is as a {\em mixture of mean-reverting processes}. Mathematically, this is analogous to the SPDE representation in Section~\ref{sec:forwardrep}. In fact, the SPDE representation and the developments here can be regarded as two instances of a single abstract {\em infinite-dimensional lift}. This unifying perspective is developed by \cite{cuchieroteichmann2018markovlifts}, but goes well beyond the scope of this chapter. Still, to emphasize the analogies {\it we will use the notation $u_t(x)$ and $\Psi(t,x)$ also in this section, though with different meanings than in Section~\ref{sec:forwardrep}}. The reader will notice strong similarities to the derivations in Section~\ref{sec:forwardrep}.

Assume that the kernel $K$ is the Laplace transform of some measure $\mu$, that is,
\begin{equation}\label{eq_Kmu}
K(t) = \int_0^\infty e^{-xt} \mu(dx), \quad t>0.
\end{equation}
If $\mu$ is a positive measure, then $K$ is completely monotone on $(0,\infty)$. Conversely, any such $K$ is of the form \eqref{eq_Kmu}, a result known as the Bernstein--Widder theorem. This clearly jibes well with Theorem~\ref{T_existence}\ref{T_existence:2}. On the other hand, $\mu$ could also be a signed measure, as long as $K$ remains in $L^2_{\rm loc}(\R_+)$. This gives a large class of kernels, not necessarily completely monotone, that are compatible with Theorem~\ref{T_existence}\ref{T_existence:1}.

\begin{example}
If we are in the classical case $K(t)=1$, then $\mu=\delta_0$. In the rough Heston case $K(t)=t^{\alpha-1}/\Gamma(\alpha)$ with $\alpha\in(\frac12,1)$, then $\mu(dx)=\frac{x^{-\alpha}}{\Gamma(\alpha)\Gamma(1-\alpha)}\,dx$.
\end{example}

To see how \eqref{eq_Kmu} leads to a (possibly infinite) mixture of mean-reverting processes, and in order to simplify the presentation, we will assume that $V_0=0$. The general case can be deduced by considering the process $\widetilde{V}=V-V_0$; the reader is invited to work out what happens in this general case. 

Substituting \eqref{eq_Kmu} into \eqref{eq_SCE} with $V_0=0$, and interchanging the time-\ and $\mu$-integrals (justified by the stochastic Fubini theorem) yields the representation
\begin{equation}\label{eq_Xux}
V_t = \int_0^\infty u_t(x) \mu(dx),
\end{equation}
where we define, for all $t\ge0$,
\begin{equation*}
u_t(x) = \int_0^t e^{-x(t-s)} b(V_s) ds + \int_0^t e^{-x(t-s)}\sigma(V_s)dW_s.
\end{equation*}
Crucially, each process $\{u_t(x)\}_{t\ge0}$ is a semimartingale, even if $V$ is not. To find its dynamics move $e^{-xt}$ outside the time integrals and apply the product rule to get
\[
du_t(x) = (-xu_t(x)+b(V_t))dt + \sigma(V_t)dW_t.
\]
Plugging \eqref{eq_b(x)a(x)} and \eqref{eq_Xux} into this expression gives
\begin{equation}\label{eq_dux}
du_t(x) = \left(\beta-xu_t(x)-\lambda \int_0^\infty u_t(y)\mu(dy) \right)dt + \sqrt{\alpha + a \int_0^\infty u_t(y)\mu(dy)} dW_t.
\end{equation}
As $x$ ranges through the support of $\mu$, \eqref{eq_dux} defines a (possibly infinite) coupled system of mean-reverting processes, and \eqref{eq_Xux} expresses $V$ as a mixture of those processes. The Gaussian case $a=0$ is covered by results of \cite{car_cou_mon_00,HARMS2018}.

Apart from its theoretical interest, this representation can be useful for numerical purposes. The idea is to replace $\mu$ by an approximation $\mu_n$ that is supported on finitely many points $x_1,\ldots,x_n$. The system \eqref{eq_dux} then becomes an SDE for the $n$-dimensional Markov process $\{u_t(x_1),\ldots,u_t(x_n)\}_{t\ge0}$. This can be used to approximate the affine Volterra process $V$. More details on this construction are given by \cite{jaberelouch2018approxmarkov} and \cite{cuchieroteichmann2018markovlifts}.

\subsection{Fourier--Laplace transforms and Riccati equations}

The drift and squared volatility in~\eqref{eq_dux} depend on the curve $u_t(\fdot)$ in an affine way. This suggests that the process $\{u_t(\fdot)\}_{t\ge0}$ is an affine Markov process, possibly infinite-dimensional. In particular, we expect a transform formula similar to \eqref{E:chfunctliftcurve}:
\begin{equation}\label{eq_LTux}
\E\left[\exp\left(\int_0^{\infty}h(x)u_T(x)\mu(dx)\right)\Mid\mathcal{F}_t\right]=\exp\left(\phi(T-t)+\int_0^{\infty}\Psi(T-t,x)u_t(x)\mu(dx)\right),
\end{equation}
where $\phi(\tau)$ and $\Psi(\tau,x)$ are solutions of appropriate Riccati equations with initial conditions
\begin{equation}\label{eq:RphiRpsi_bdry}
\phi(0) = 0, \quad \Psi(0,x) = h(x).
\end{equation}
In this Markovian situation, one can apply the standard method for deriving the Riccati equations.
Let $M_t$ denote the right-hand side of \eqref{eq_LTux}. It\^o's formula and \eqref{eq_dux} give, after some computations,
\begin{equation}\label{eq_dM/M_Laplace}
\begin{aligned}
\frac{dM_t}{M_t} &= \Big[ -\partial_t\phi(T-t) + \Rcal_{\phi}\big(\textstyle{\int}_0^{\infty}\Psi(T-t,y) \mu(dy)\big) \\
&\qquad + \displaystyle{\int}_0^\infty \Big( - \partial_t\Psi(T-t,x) - x\Psi(T-t,x) \\
&\qquad\qquad\qquad + \Rcal_{\Psi}\big(\textstyle{\int}_0^{\infty}\Psi(T-t,y) \mu(dy)\big)\Big) u_t(x) \mu(dx) \Big] dt \\
&\quad + \text{local martingale},
\end{aligned}
\end{equation}
with $\mathcal R_{\phi}$, $\mathcal R_{\Psi}$ as in~\eqref{eq:RphiRpsi}. It is remarkable that the same functions $\mathcal R_{\phi}$ and $\mathcal R_{\Psi}$ as for the SPDE representation occur also here. This is one manifestation of the underlying abstract point of view due to \cite{cuchieroteichmann2018markovlifts}.

Suppose that $\phi$ and $\Psi$ solve the possibly infinite-dimensional Riccati equations
\begin{align}
\partial_t\phi(t)&= \Rcal_{\phi}\left(\textstyle\int_0^{\infty}\Psi(t,y)\mu(dy)\right),\nonumber\\
\partial_t\Psi(t,x)&=-x\Psi(t,x)+\mathcal R_{\Psi}\left(\textstyle\int_0^{\infty}\Psi(t,y)\mu(dy)\right), \label{E:Riccatilift2_measure}
\end{align}
with initial conditions \eqref{eq:RphiRpsi_bdry}. Then, due to \eqref{eq_dM/M_Laplace}, $M$ is a local martingale with $M_T=\exp(\int_0^{\infty}h(x)u_T(x)\mu(dx))$. If $M$ is actually a true martingale, we obtain the transform formula \eqref{eq_LTux}, which is nothing but the martingale property $\E[M_T\mid\Fcal_t]=M_t$. In particular, if $h(x)\equiv v$ is constant, combining \eqref{eq_Xux} and \eqref{eq_LTux} gives
\[
\E\left[\exp\left( v V_T \right)\mid\mathcal{F}_t\right]=\exp\left(\phi(T-t)+\int_0^{\infty}\Psi(T-t,x)u_t(x)\mu(dx)\right).
\]

Just as in Section~\ref{sec:forwardrep}, there is a connection between the solution $\Psi(t,x)$ to the Riccati equation \eqref{E:Riccatilift2_measure}, with $h(x)\equiv v$ constant, and the solution $\psi(t)$ to the Riccati--Volterra equation \eqref{eq:Riccati_kernel}, with $u=w=0$. The link is given by the formula
\[
\psi(t)=\int_0^{\infty} \Psi(t,x)\mu(dx),
\]
which can be verified by similar calculations are in Section~\ref{sec:forwardrep}. This gives yet another way to derive the Fourier--Laplace transform formula.

\bibliographystyle{plainnat}
\bibliography{references}

\begin{thebibliography}{20}
\providecommand{\natexlab}[1]{#1}
\providecommand{\url}[1]{\texttt{#1}}
\expandafter\ifx\csname urlstyle\endcsname\relax
  \providecommand{\doi}[1]{doi: #1}\else
  \providecommand{\doi}{doi: \begingroup \urlstyle{rm}\Url}\fi

\bibitem[Abi~Jaber and
  El~Euch(2018{\natexlab{a}})]{jaberelouch2018approxmarkov}
Eduardo Abi~Jaber and Omar El~Euch.
\newblock Multi-factor approximation of rough volatility models.
\newblock \emph{arXiv:1801.10359}, 2018{\natexlab{a}}.

\bibitem[Abi~Jaber and
  El~Euch(2018{\natexlab{b}})]{jaberelouch2018markovaffineroughHeston}
Eduardo Abi~Jaber and Omar El~Euch.
\newblock Markovian structure of the {V}olterra {H}eston model.
\newblock \emph{arXiv:1803.00477}, 2018{\natexlab{b}}.

\bibitem[Abi~Jaber et~al.(2017)Abi~Jaber, Larsson, and Pulido]{jaber2017affine}
Eduardo Abi~Jaber, Martin Larsson, and Sergio Pulido.
\newblock Affine {V}olterra processes.
\newblock \emph{arXiv:1708.08796}, 2017.

\bibitem[Bayer et~al.(2016)Bayer, Friz, and Gatheral]{bay_fri_gat_16}
Christian Bayer, Peter Friz, and Jim Gatheral.
\newblock Pricing under rough volatility.
\newblock \emph{Quantitative Finance}, 16\penalty0 (6):\penalty0 887--904,
  2016.

\bibitem[Bergomi and Guyon(2012)]{ber_guy_12}
L.~Bergomi and J.~Guyon.
\newblock Stochastic volatility's orderly smiles.
\newblock \emph{Risk}, 25\penalty0 (5):\penalty0 60--66, 2012.

\bibitem[B{\"u}hler(2006)]{buehler2006volatility}
Hans B{\"u}hler.
\newblock \emph{Volatility Markets -- Consistent modeling, hedging and
  practical implementation}.
\newblock PhD thesis, TU Berlin, 2006.

\bibitem[Carmona et~al.(2000)Carmona, Coutin, and Montseny]{car_cou_mon_00}
Philippe Carmona, Laure Coutin, and G.~Montseny.
\newblock Approximation of some {G}aussian processes.
\newblock \emph{Stat. Inference Stoch. Process.}, 3\penalty0 (1-2):\penalty0
  161--171, 2000.
\newblock 19th ``Rencontres Franco-Belges de Statisticiens'' (Marseille, 1998).

\bibitem[Comte et~al.(2012)Comte, Coutin, and Renault]{Comte2012}
F.~Comte, L.~Coutin, and E.~Renault.
\newblock Affine fractional stochastic volatility models.
\newblock \emph{Annals of Finance}, 8\penalty0 (2):\penalty0 337--378, May
  2012.

\bibitem[Cuchiero and Teichmann(2018)]{cuchieroteichmann2018markovlifts}
Christa Cuchiero and Josef Teichmann.
\newblock Generalized {F}eller processes and {M}arkovian lifts of stochastic
  {V}olterra processes: the affine case.
\newblock \emph{arXiv:1804.10450}, 2018.

\bibitem[Da~Prato and Zabczyk(2014)]{da2014stochastic}
Giuseppe Da~Prato and Jerzy Zabczyk.
\newblock \emph{Stochastic equations in infinite dimensions}.
\newblock Cambridge university press, 2014.

\bibitem[Duffie et~al.(2003)Duffie, Filipovi\'c, and
  Schachermayer]{duf_fil_sch_03}
D.~Duffie, D.~Filipovi\'c, and W.~Schachermayer.
\newblock Affine processes and applications in finance.
\newblock \emph{Ann. Appl. Probab.}, 13\penalty0 (3):\penalty0 984--1053, 2003.

\bibitem[El~Euch and Rosenbaum(2016)]{euch2016characteristic}
Omar El~Euch and Mathieu Rosenbaum.
\newblock The characteristic function of rough {H}eston models.
\newblock \emph{Mathematical Finance}, 2016.

\bibitem[Fukasawa(2017)]{fukasawa2017short}
Masaaki Fukasawa.
\newblock Short-time at-the-money skew and rough fractional volatility.
\newblock \emph{Quantitative Finance}, 17\penalty0 (2):\penalty0 189--198,
  2017.

\bibitem[Gatheral and Keller-Ressel(2018)]{gatheral2018affine}
Jim Gatheral and Martin Keller-Ressel.
\newblock Affine forward variance models.
\newblock \emph{arXiv:1801.06416}, 2018.

\bibitem[Gatheral et~al.(2018)Gatheral, Jaisson, and
  Rosenbaum]{gatheral2018volatility}
Jim Gatheral, Thibault Jaisson, and Mathieu Rosenbaum.
\newblock Volatility is rough.
\newblock \emph{Quantitative Finance}, 18\penalty0 (6):\penalty0 933--949,
  2018.

\bibitem[Guennoun et~al.(2018)Guennoun, Jacquier, Roome, and Shi]{GJR:14}
Hamza Guennoun, Antoine Jacquier, Patrick Roome, and Fangwei Shi.
\newblock Asymptotic behavior of the fractional {H}eston model.
\newblock \emph{SIAM Journal on Financial Mathematics}, 9\penalty0
  (3):\penalty0 1017--1045, 2018.

\bibitem[Harms and Stefanovits(2018)]{HARMS2018}
Philipp Harms and David Stefanovits.
\newblock Affine representations of fractional processes with applications in
  mathematical finance.
\newblock \emph{Stochastic Processes and their Applications}, 2018.
\newblock In press. doi:10.1016/j.spa.2018.04.010.

\bibitem[Heston(1993)]{heston1993closed}
Steven~L Heston.
\newblock A closed-form solution for options with stochastic volatility with
  applications to bond and currency options.
\newblock \emph{The Review of Financial Studies}, 6\penalty0 (2):\penalty0
  327--343, 1993.

\bibitem[Kallsen(2006)]{Kallsen:2006}
J.~Kallsen.
\newblock A didactic note on affine stochastic volatility models.
\newblock In \emph{From stochastic calculus to mathematical finance}, pages
  343--368. Springer, 2006.

\bibitem[Veraar(2012)]{veraar2012stochastic}
Mark Veraar.
\newblock The stochastic {F}ubini theorem revisited.
\newblock \emph{Stochastics An International Journal of Probability and
  Stochastic Processes}, 84\penalty0 (4):\penalty0 543--551, 2012.

\end{thebibliography}
\end{document}